%% file: main.tex
\documentclass[letterpaper,10 pt,conference]{ieeeconf}
\IEEEoverridecommandlockouts
\usepackage[utf8]{inputenc}
\usepackage{graphicx}
\usepackage[table]{xcolor}
\usepackage{hhline}
 
\usepackage{amsmath,amsthm}
\usepackage{mathtools}
\usepackage{booktabs}
\usepackage{dsfont}
\usepackage[linesnumbered,ruled,vlined]{algorithm2e}
\DontPrintSemicolon
\usepackage[style=ieee,sorting=none,maxbibnames=99,giveninits,doi=false]{biblatex}
\input{PackagesCommands}

\title{\LARGE \bf Discrete-Time Matrix Control Barrier Functions \\ 
for Disjunctive Safe Sets}

\title{\bf Convex Safety Filters for Nonconvex Safe Sets via \\ Discrete-Time Matrix Control Barrier Functions}

\title{Convex Safety Filtering via Spectral Selection for Nonconvex Safe Sets}

\author{Juan Augusto Paredes Salazar, James Usevitch, and Ankit Goel%
\thanks{J. A. Paredes Salazar and A. Goel are with the Department of Mechanical Engineering, University of Maryland, Baltimore County, Baltimore, MD 21250. {\tt\small japarede, ankgoel@umbc.edu}}%
\thanks{J. A. Paredes Salazar is now with the Department of Aerospace Engineering, The University of Michigan, Ann Arbor, MI 48109. {\tt\small jparedes@umich.edu}}%
\thanks{J. Usevitch is with the Department of Electrical and Computer Engineering, Brigham Young University, Provo, UT 84602. {\tt\small james\_usevitch@byu.edu}}%
}

\begin{document}
\maketitle
\thispagestyle{empty}
\pagestyle{empty}

\begin{abstract}
The discrete-time control barrier function condition for a safe set that is a union of convex sets is nonconvex in the control input.
For safe sets defined by a matrix concave function through the number of its nonnegative eigenvalues, this paper shows that the set is a union of convex sets, and that the nonconvexity arises because at least one member of the union must hold, not every member.
Selecting eigenvectors of the matrix function at the current state yields a convex input constraint that acts only on the eigenvalues determining membership in the set.
This constraint is implied by the matrix-wide condition of prior work, and it guarantees a geometric lower bound on each of those eigenvalues along closed-loop trajectories.
A double-integrator simulation with a polytope obstacle and a spectrahedron obstacle compares the proposed filter with the matrix-wide condition.
\end{abstract}

\section{Introduction}
\label{sec:intro}

Autonomous vehicles must satisfy safety constraints while executing their assigned tasks.
For autonomous cars, these constraints include safe separation from other vehicles, adherence to lane boundaries, and avoidance of pedestrians who enter the roadway.
For unmanned aerial vehicles, they include safe distances from obstacles and other vehicles and adherence to admissible flight regions.
Safety filters enforce such constraints by modifying the input of a nominal controller only when necessary to keep the state in a safe set \cite{wabersich2021filter,hsu2024filter}.
Filters based on control barrier functions (CBFs) provide formal safety guarantees while minimally modifying the nominal controller \cite{ames2017qp}, and have been applied to adaptive cruise control \cite{ames2014acc}, combined lane keeping and cruise control \cite{xu2018composition}, multicopter obstacle avoidance \cite{wu2016quadrotor}, and fixed-wing collision avoidance \cite{squires2018aircraft}.

CBFs extend barrier certificates for safety verification \cite{prajna2004certificates} to feedback design \cite{wieland2007constructive}.
Safety is formalized as forward invariance of the safe set.
For continuous-time, control-affine systems, a differentiable scalar barrier yields a constraint that is affine in the input, so the safety filter is a convex quadratic program whenever the remaining input constraints are convex, whether or not the safe set is convex \cite{ames2017qp}.
Digital controllers measure and actuate at sampling times, and continuous-time guarantees do not carry over directly to sampled implementations.
Discrete-time CBFs account for the state transition of each update and have been applied to bipedal robot navigation \cite{agrawal2017discrete} and to model predictive control for obstacle avoidance and car racing \cite{zeng2021mpccbf}.
Sampled-data conditions \cite{breeden2022sampled} and analyses based on approximate discrete-time models \cite{taylor2022approximate} connect digital implementation with continuous-time safety.

Discrete-time CBF constraints evaluate the barrier at the next state, and this composition is generally nonlinear in the input even for control-affine dynamics.
The constraint is convex in the input when the composition is concave, but common obstacle-avoidance barriers are not concave \cite{agrawal2017discrete,zeng2021mpccbf}.
Requiring the next position to lie outside a circular or ellipsoidal obstacle, for example, gives a reverse-convex constraint.
Convexity of the constraint in the input at a fixed state is the relevant property, and a nonconvex safe set does not by itself imply a nonconvex constraint.
A nonconvex filter requires either a global solver, with long solve times, or a local solver, which can return inputs that violate the constraint.
This paper considers safe sets that are unions of convex sets, such as the complement of a polytope, the complement of a spectrahedron, and free space described by convex cells.

Several works obtain convex discrete-time filters under additional structure.
Partially control-affine discrete-time CBFs yield affine input conditions through mixed-integer optimization \cite{khajenejad2021tractable}.
Duality-based filtering handles polytopic obstacles in continuous time \cite{thirugnanam2022duality}, and supporting hyperplanes yield convex programs within iterative discrete-time model predictive control \cite{liu2026iterative}.
The latter considers scalar-valued safe set functions and does not treat disjunctions of multiple CBFs.
These formulations assume specific dynamics or obstacle geometry, and some restrict the admissible states or inputs, which can introduce conservatism.

Matrix control barrier functions (MCBFs) represent safe sets through symmetric matrix-valued functions of the state.
In continuous time, they accommodate semidefinite and indefinite conditions, eigenvalue constraints, and conjunctions and disjunctions of scalar constraints, with filters posed as semidefinite programs \cite{ong2026matrix}.
Extensions cover combinatorial compositions \cite{ong2025combinatorial}, sampled-data robustness \cite{usevitch2026sampledmatrix}, and higher-order constraints \cite{gessow2026highordermatrix}.
The work \cite{usevitch2026dtmcbf} introduces discrete-time MCBFs and gives two convex constructions.
The first applies to safe sets on which a prescribed number of eigenvalues of the matrix function are nonnegative, which includes disjunctions of scalar constraints.
It bounds the decrease of the entire matrix at each step, and the resulting constraint is convex in the input when the matrix function is matrix concave.
The second projects the current state onto each convex component of the unsafe set and forms separating hyperplanes.
It requires the unsafe set to lie in a finite union of convex sets and solves one convex program per obstacle at each step, which is a semidefinite program for a spectrahedron.
The first construction also constrains eigenvalues that do not affect membership in the set, and the second requires a projection at each step.

This paper constructs a convex constraint that acts only on the eigenvalues that determine membership in the set.
When the matrix function is matrix concave, the safe set is a union of convex sets, and the discrete-time constraint is nonconvex because at least one member of the union must hold, not every member.
Selecting a member through eigenvectors of the matrix function at the current state gives a convex constraint on the next state.
The main contributions of the paper are as follows.
\begin{enumerate}

    \item The safe set is shown to be a union of convex sets when the matrix function is matrix concave (Proposition \ref{prop:union}).
    Complements of polytopes and spectrahedra, which are written as such sets in \cite{usevitch2026dtmcbf}, and unions of convex sets are covered (Section \ref{sec:obstacles}).
    \item A convex input constraint is constructed by selecting eigenvectors at the current state, for any number of required nonnegative eigenvalues (Section \ref{sec:main_results}).
    Each step requires one eigendecomposition and one convex program, and the construction does not use projection onto the unsafe set.
    \item Every eigenvalue that determines membership in the set is shown to satisfy a geometric lower bound along closed-loop trajectories (Theorem \ref{thm:main}).
    \item The condition of \cite{usevitch2026dtmcbf} that bounds the entire matrix is shown to imply the proposed constraint, so the proposed constraint admits at least as many inputs at each state (Proposition \ref{prop:weyl_condition}).
\end{enumerate}
%


The paper is organized as follows.
Section \ref{sec:prelim} gives notation, matrix concavity, and the discrete-time MCBF framework of \cite{usevitch2026dtmcbf}.
Section \ref{sec:obstacles} shows that complements of polytopes, complements of spectrahedra, and unions of convex sets share a common matrix-valued representation.
Section \ref{sec:main_results} presents the union characterization, the eigenvector selection constraint and its properties, and the comparison with the full-matrix condition.
Section \ref{sec:simulations} compares the two filters on a double integrator with two obstacles.
Section \ref{sec:conclusions} concludes the paper.

\section{Preliminaries}
\label{sec:prelim}

\subsection{Notation}

The set of symmetric real matrices of size \(p\times p\) is denoted \(\mathbb{S}^p\).
For \(A, B \in \mathbb{S}^p\), \(A \succeq B\) indicates that \(A - B\) is positive semidefinite, and \(A \succ B\) indicates that \(A - B\) is positive definite.
The eigenvalues of \(A \in \mathbb{S}^p\) are ordered as \(\lambda_1(A) \leq \lambda_2(A) \leq \cdots \leq \lambda_p(A)\).
The identity matrix of size \(r \times r\) is denoted \(I_r\).
A matrix \(Y \in \R^{p \times r}\) with \(Y^\intercal Y = I_r\) is called orthonormal.

The following two eigenvalue inequalities are used in the sequel.

\begin{lemma}[{Weyl \cite[Thm. 4.3.1]{horn2012matrix}}]
    \label{lem:weyl}
    Let \(A, B \in \mathbb{S}^p\).
    Then, for all \(i \in \{1,\ldots,p\}\),
    \begin{align}
        \lambda_i(A) + \lambda_1(B) \leq \lambda_i(A+B). \label{eq:weyl}
    \end{align}
\end{lemma}

\begin{lemma}[{Cauchy interlacing \cite[Thm. 4.3.28]{horn2012matrix}}]
    \label{lem:interlacing}
    Let \(A \in \mathbb{S}^p\) and let \(Y \in \R^{p \times r}\) be orthonormal, where \(1 \leq r \leq p\).
    Then, for all \(i \in \{1,\ldots,r\}\),
    \begin{align}
        \lambda_i(A) \leq \lambda_i(Y^\intercal A Y) \leq \lambda_{i+p-r}(A). \label{eq:interlacing}
    \end{align}
\end{lemma}

\subsection{Matrix Concavity}

\begin{define}[{\cite[Sec. 3.6.2]{boyd2004convex}}]
    \label{def:matrix_concavity}
    A mapping \(H : \R^n \to \mathbb{S}^p\) is matrix concave if, for all \(x, y \in \R^n\) and all \(\theta \in [0,1]\),
    \begin{align}
        H(\theta x + (1-\theta) y) \succeq \theta H(x) + (1-\theta) H(y).
    \end{align}
    Equivalently, \(H\) is matrix concave if and only if \(x \mapsto v^\intercal H(x) v\) is concave for all \(v \in \R^p\).
\end{define}

Two consequences of Definition \ref{def:matrix_concavity} are used below.
First, if \(H\) is matrix concave and \(A \in \R^{n \times m}\), \(b \in \R^n\), then \(u \mapsto H(A u + b)\) is matrix concave, since the composition of a concave function with an affine mapping is concave \cite[Sec. 3.2.2]{boyd2004convex}.
Second, if \(H\) is matrix concave and \(Y \in \R^{p \times r}\), then \(x \mapsto Y^\intercal H(x) Y\) is matrix concave, since \(w^\intercal Y^\intercal H(x) Y w = v^\intercal H(x) v\) with \(v = Y w\).

\subsection{Discrete-Time Matrix Control Barrier Functions}

Consider the discrete-time control-affine system
\begin{align}
    x_{k+1} = f(x_k) + g(x_k) u_k, \label{eq:dynamics}
\end{align}
where \(x_k \in \R^n\) is the state, \(u_k \in \R^m\) is the control input, and \(f : \R^n \to \R^n\) and \(g : \R^n \to \R^{n \times m}\) are locally Lipschitz.

\begin{define}
    \label{def:forward_invariance}
    A set \(\SC \subset \R^n\) is forward invariant with respect to \eqref{eq:dynamics} if \(x_0 \in \SC\) implies \(x_k \in \SC\) for all \(k \geq 0\).
\end{define}

Given \(H : \R^n \to \mathbb{S}^p\), the work \cite{usevitch2026dtmcbf} considers safe sets of the form
\begin{align}
    \SSS = \{ x \in \R^n : H(x) \succeq 0 \}, \label{eq:matrix_safe_set}
\end{align}
and shows that \(\SSS\) is rendered forward invariant by any control law satisfying \(u_k \in \mathcal{K}_H(x_k)\) for all \(k \geq 0\), where
\small
\begin{align}
    \mathcal{K}_H(x) \triangleq \{ u \in \R^m : H(f(x) + g(x) u) \succeq (1-\gamma) H(x) \} \label{eq:K_H}
\end{align}
\normalsize
and \(0 < \gamma \leq 1\).
Given a nominal control law \(u_{\rm nom} : \R^n \to \R^m\), a safe control input is computed by
\small
\begin{subequations}
\label{eq:safety_filter}
\begin{align}
    u^*(x_k) = \underset{u \in \R^m}{\arg\min} \ & \nrm{u - u_{\rm nom}(x_k)}_2 \\
    \text{s.t.} \ & H(f(x_k) + g(x_k) u) \succeq (1-\gamma) H(x_k). \label{eq:safety_filter_constraint}
\end{align}
\end{subequations}
\normalsize
If \(H\) is matrix concave, then \eqref{eq:safety_filter_constraint} is a linear matrix inequality in \(u\) and \eqref{eq:safety_filter} is convex.
If \(H\) is not matrix concave, then \eqref{eq:safety_filter} is nonconvex in general.

\subsection{Safe Subset Functions}

To handle safe sets \eqref{eq:matrix_safe_set} defined by a nonconcave \(H\), the work \cite{usevitch2026dtmcbf} replaces \(H\) by a two-argument function whose first argument is the point whose safety is tested and whose second argument selects a convex subset of \(\SSS\).

\begin{define}[{\cite{usevitch2026dtmcbf}}]
    \label{def:safe_subset_function}
    Let \(\SSS \subset \R^n\) and let \(\mathcal{D} \supseteq \SSS\).
    A function \(\Phi : \R^n \times \R^n \to \mathbb{S}^q\) is a \emph{safe subset function} for \(\SSS\) on \(\mathcal{D}\) if
    \begin{align}
        \{ z \in \R^n : \Phi(z, x) \succeq 0 \} \subseteq \SSS \quad \forall x \in \mathcal{D}, \label{eq:subset_property}
    \end{align}
    and \(\Phi(x, x) \succeq 0\) for all \(x \in \SSS\).
\end{define}

\begin{define}[{\cite{usevitch2026dtmcbf}}]
    \label{def:SDTE-MCBF}
    A safe subset function \(\Phi\) for \(\SSS\) on \(\mathcal{D}\) is a \emph{subset-based discrete-time exponential matrix control barrier function} (SDTE-MCBF) for \eqref{eq:dynamics} if there exists \(0 < \gamma \leq 1\) such that, for all \(x \in \mathcal{D}\), there exists \(u \in \R^m\) satisfying
    \begin{align}
        \Phi(f(x) + g(x) u, x) \succeq (1-\gamma) \Phi(x, x). \label{eq:SDTE_condition}
    \end{align}
\end{define}

Given an SDTE-MCBF \(\Phi\), define
\footnotesize
\begin{align}
    \mathcal{K}_\Phi(x) \triangleq \{ u \in \R^m : \Phi(f(x) + g(x) u, x) \succeq (1-\gamma) \Phi(x, x) \}. \label{eq:K_Phi}
\end{align}
\normalsize

\begin{theorem}[{\cite[Thm. X]{usevitch2026dtmcbf}}]
    \label{thm:SDTE_invariance}
    Let \(\Phi\) be an SDTE-MCBF for \(\SSS\) on \(\mathcal{D}\).
    Then any control law satisfying \(u_k \in \mathcal{K}_\Phi(x_k)\) for all \(k \geq 0\) renders \(\SSS\) forward invariant.
\end{theorem}

In \cite{usevitch2026dtmcbf}, the subset in \eqref{eq:subset_property} is constructed by projecting \(x\) onto each convex component of the unsafe set and forming a separating hyperplane.
That construction requires the unsafe set to be contained in a finite union of closed convex sets, so that the safe set is an intersection of the complements.
The present paper considers the complementary case, in which the safe set is itself a union of convex sets, and constructs the subset in \eqref{eq:subset_property} by a different mechanism.

\section{Nonconvex Safe Sets Arising from Convex Obstacles}
\label{sec:obstacles}

\begin{figure}[h]
    \vspace{-0.5em}
    \centering
    \input{Figures/fig_obstacles_body}
    \caption{
    Nonconvex safe sets arising from convex obstacles.
    (A) The complement of a polytope is a finite union of halfspaces; shading depth indicates the number of halfspaces containing a point.
    (B) The complement of a spectrahedron is a union of halfspaces \(\{x : v^\intercal H(x) v \geq 0\}\) indexed by the unit sphere; a few members are shown dashed.
    (C) A union of convex cells \(\SSS_1 \cup \SSS_2 \cup \SSS_3\).
    }
    \label{fig:obstacles}
    \vspace{-1.5em}
\end{figure}
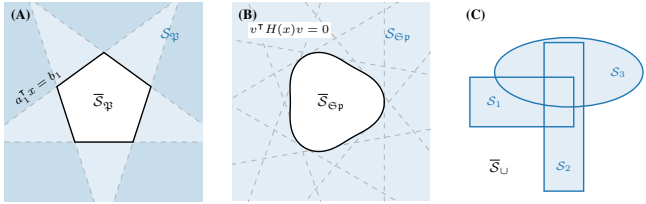

This section presents three classes of safe sets, shown in Figure \ref{fig:obstacles}, and writes each in the form
\begin{align}
    \SSS = \{ x \in \R^n : \lambda_p(H(x)) \geq 0 \} \label{eq:lambda_max_set}
\end{align}
for a matrix concave \(H : \R^n \to \mathbb{S}^p\).
In each case \(\SSS\) is nonconvex and the constraint obtained by replacing \(H \succeq 0\) in \eqref{eq:safety_filter_constraint} with \(\lambda_p(H) \geq 0\) is nonconvex in \(u\).
Sets of the form \eqref{eq:lambda_max_set} are treated in Section \ref{sec:main_results}.

\subsection{Complement of a Polytope}
\label{ssec:polytope}

Let \(\overline{\SSS}_\mathrm{P} = \{ x \in \R^n : a_i^\intercal x \leq b_i,\ i \in \{1,\ldots,p\} \}\) with \(a_i \in \R^n\), \(b_i \in \R\), and let the safe set be its complement, Figure \ref{fig:obstacles}(A).
Define
\begin{align}
    H_\mathrm{P}(x) \triangleq \operatorname{diag}\!\pth{a_1^\intercal x - b_1, \ldots, a_p^\intercal x - b_p}. \label{eq:H_polytope}
\end{align}
Then \(H_\mathrm{P}\) is affine, hence matrix concave, \(\lambda_p(H_\mathrm{P}(x)) = \max_i (a_i^\intercal x - b_i)\), and
\small
\begin{align}
    \SSS_\mathrm{P} = \{ x \in \R^n : \lambda_p(H_\mathrm{P}(x)) \geq 0 \} = \bigcup_{i=1}^p \{ x : a_i^\intercal x \geq b_i \}. \label{eq:polytope_safe}
\end{align}
\normalsize
%
The constraint \(\max_i \pth{a_i^\intercal (f(x_k) + g(x_k) u) - b_i} \geq (1-\gamma) \lambda_p(H_\mathrm{P}(x_k))\) confines \(u\) to a union of \(p\) halfspaces in \(\R^m\).
Enforcing all \(p\) inequalities simultaneously would be convex but corresponds to \(\{x : H_\mathrm{P}(x) \succeq 0\}\), which is empty when \(\overline{\SSS}_\mathrm{P}\) is bounded.

\subsection{Complement of a Spectrahedron}
\label{ssec:spectrahedron}

Let \(A_0, A_1, \ldots, A_n \in \mathbb{S}^p\) and let \(\overline{\SSS}_\mathrm{Sp} = \{ x \in \R^n : A_0 + \sum_{i=1}^n x_i A_i \succeq 0 \}\), a spectrahedron \cite{ramana1995some}, be the obstacle, Figure \ref{fig:obstacles}(B).
%
Define
\begin{align}
    H_\mathrm{Sp}(x) \triangleq -\Big( A_0 + \sum_{i=1}^n x_i A_i \Big). \label{eq:H_spectrahedron}
\end{align}
Then \(H_\mathrm{Sp}\) is affine, and \(x \notin \operatorname{int}(\overline{\SSS}_\mathrm{Sp})\) if and only if \(H_\mathrm{Sp}(x) \not\prec 0\), so that
\begin{align}
    \SSS_\mathrm{Sp} 
        &=
            \{ x \in \R^n : \lambda_p(H_\mathrm{Sp}(x)) \geq 0 \} 
        \nn \\
        &= \bigcup_{\nrm{v}_2 = 1} \{ x : v^\intercal H_\mathrm{Sp}(x) v \geq 0 \}, \label{eq:spectrahedron_safe}
\end{align}
where the second equality follows from \(\lambda_p(A) = \max_{\nrm{v}_2 = 1} v^\intercal A v\).
Each set in the union is a halfspace, and the polytope of Section \ref{ssec:polytope} is the special case of diagonal \(A_i\), for which the union is finite.
Since \(\lambda_p\) of an affine matrix function is convex \cite[Ex. 3.10]{boyd2004convex}, the constraint \(\lambda_p(H_\mathrm{Sp}(f(x_k) + g(x_k) u)) \geq (1-\gamma) \lambda_p(H_\mathrm{Sp}(x_k))\) is the superlevel set of a convex function of \(u\), which is nonconvex.

\subsection{Union of Convex Sets}
\label{ssec:union}

Let \(\SSS_i = \{ x \in \R^n : H_i(x) \succeq 0 \}\), \(i \in \{1,\ldots,Q\}\), where each \(H_i : \R^n \to \mathbb{S}^{p_i}\) is matrix concave, and let the safe set be \(\SSS_\cup = \bigcup_{i=1}^Q \SSS_i\), Figure \ref{fig:obstacles}(C).
This case arises when the free space is described by convex cells or corridors.

The block-diagonal function \(\operatorname{diag}(H_1(x), \ldots, H_Q(x))\) does not encode \(\SSS_\cup\) through \eqref{eq:lambda_max_set}, since its largest eigenvalue is nonnegative whenever some \(H_i(x)\) has a nonnegative eigenvalue, which does not imply \(H_i(x) \succeq 0\).
Each \(\SSS_i\) is instead reduced to a scalar inequality.

\begin{lemma}
    \label{lem:lambda_min_concave}
    If \(H_i\) is matrix concave, then \(h_i(x) \triangleq \lambda_1(H_i(x))\) is concave and \(\SSS_i = \{ x : h_i(x) \geq 0 \}\).
\end{lemma}
\begin{proof}
    Since \(h_i(x) = \min_{\nrm{v}_2 = 1} v^\intercal H_i(x) v\) is a pointwise minimum of concave functions of \(x\), it is concave \cite[Sec. 3.2.3]{boyd2004convex}.
    The set identity holds since \(H_i(x) \succeq 0\) if and only if \(\lambda_1(H_i(x)) \geq 0\).
\end{proof}

Define
\begin{align}
    H_\cup(x) \triangleq \operatorname{diag}\!\pth{h_1(x), \ldots, h_Q(x)}. \label{eq:H_union}
\end{align}
By Lemma \ref{lem:lambda_min_concave}, \(H_\cup\) is matrix concave, \(\lambda_Q(H_\cup(x)) = \max_i h_i(x)\), and \(\SSS_\cup = \{ x : \lambda_Q(H_\cup(x)) \geq 0 \}\), which is of the form \eqref{eq:lambda_max_set} with \(p = Q\).
The constraint \(\max_i h_i(f(x_k) + g(x_k) u) \geq (1-\gamma) \max_i h_i(x_k)\) confines \(u\) to a union of \(Q\) convex sets.

\subsection{Summary}

In all three cases the safe set is a union of convex sets, each of the form \(\{ x : v^\intercal H(x) v \geq 0 \}\) for some \(v \in \R^p\) and matrix concave \(H\), and the safety filter is nonconvex only because at least one, rather than every, such inequality must hold.
Section \ref{sec:main_results} shows that selecting \(v\) from the current state at each step yields a convex constraint that preserves forward invariance.

\section{Disjunctive Safe Sets and Convex Safety Filters}
\label{sec:main_results}


\subsection{Disjunctive Safe Sets}

Let \(H : \R^n \to \mathbb{S}^p\) and \(j \in \{1,\ldots,p\}\), and define
\begin{align}
    \SSS^{(j)} \triangleq \{ x \in \R^n : \lambda_j(H(x)) \geq 0 \}. \label{eq:S_j}
\end{align}
Since the eigenvalues are ordered, \(x \in \SSS^{(j)}\) if and only if at least \(r \triangleq p - j + 1\) eigenvalues of \(H(x)\) are nonnegative.
For \(j = 1\), \eqref{eq:S_j} reduces to the conjunctive set \eqref{eq:matrix_safe_set}, and for \(j = p\) to the set \eqref{eq:lambda_max_set} of Section \ref{sec:obstacles}.
For diagonal \(H = \operatorname{diag}(h_1, \ldots, h_p)\), \(\SSS^{(j)}\) is the set on which at least \(r\) of the \(p\) scalar inequalities \(h_i(x) \geq 0\) hold, which is the discrete-time counterpart of the \(p\)-choose-\(r\) constraints of \cite{ong2025combinatorial}.

The next result shows that \(\SSS^{(j)}\) is a union of convex sets whenever \(H\) is matrix concave, extending the observation of Section \ref{sec:obstacles} to \(j < p\).
Let \(\mathcal{O}_{p,r}\) denote the set of orthonormal matrices \(Y \in \R^{p \times r}\).

\begin{proposition}
    \label{prop:union}
    Let \(H\) be matrix concave.
    Then
    \begin{align}
        \SSS^{(j)} = \bigcup_{Y \in \mathcal{O}_{p,r}} \{ x \in \R^n : Y^\intercal H(x) Y \succeq 0 \}, \label{eq:S_j_union}
    \end{align}
    and each set in the union is convex.
\end{proposition}
\begin{proof}
    By the Courant--Fischer theorem \cite[Thm. 4.2.6]{horn2012matrix}, \(\lambda_j(A) = \max_{Y \in \mathcal{O}_{p,r}} \lambda_1(Y^\intercal A Y)\) for \(A \in \mathbb{S}^p\), which gives \eqref{eq:S_j_union}.
    Each \(x \mapsto Y^\intercal H(x) Y\) is matrix concave by Section \ref{sec:prelim}, so its superlevel set is convex.
\end{proof}

\subsection{Eigenvector Selection}

Proposition \ref{prop:union} suggests fixing \(Y\) at each step and enforcing the corresponding convex inequality.
The natural choice is the \(Y\) that attains the maximum in the Courant--Fischer characterization at the current state.
For \(x \in \R^n\), let \(Y(x) \in \mathcal{O}_{p,r}\) have columns that are eigenvectors of \(H(x)\) associated with \(\lambda_j(H(x)), \ldots, \lambda_p(H(x))\), and let
\begin{equation}
    \resizebox{\columnwidth}{!}{$\displaystyle
    \Lambda(x) \triangleq Y(x)^\intercal H(x) Y(x) = \operatorname{diag}\!\pth{\lambda_j(H(x)), \ldots, \lambda_p(H(x))}
    $}
    \label{eq:Lambda}
\end{equation}
%
If \(\lambda_{j-1}(H(x)) = \lambda_j(H(x))\), then \(Y(x)\) is not unique; any choice satisfying \eqref{eq:Lambda} is admissible below.
Define \(\Phi_H : \R^n \times \R^n \to \mathbb{S}^r\) by
\begin{align}
    \Phi_H(z, x) \triangleq Y(x)^\intercal H(z)\, Y(x). \label{eq:Phi_H}
\end{align}

\begin{lemma}
    \label{lem:Phi_safe_subset}
    \(\Phi_H\) is a safe subset function for \(\SSS^{(j)}\) on \(\mathcal{D} = \R^n\).
\end{lemma}
\begin{proof}
    Let \(x \in \R^n\) and \(z\) satisfy \(\Phi_H(z, x) \succeq 0\), that is, \(\lambda_1(Y(x)^\intercal H(z) Y(x)) \geq 0\).
    By Lemma \ref{lem:interlacing} with \(i = 1\), \(\lambda_1(Y(x)^\intercal H(z) Y(x)) \leq \lambda_{1+p-r}(H(z)) = \lambda_j(H(z))\), so \(z \in \SSS^{(j)}\), which verifies \eqref{eq:subset_property}.
    Finally, \(\Phi_H(x, x) = \Lambda(x) \succeq 0\) for all \(x \in \SSS^{(j)}\) by \eqref{eq:Lambda}.
\end{proof}

The constraint \eqref{eq:SDTE_condition} applied to \(\Phi_H\) reads
\begin{align}
    Y(x_k)^\intercal H(f(x_k) + g(x_k) u)\, Y(x_k) \succeq (1-\gamma)\, \Lambda(x_k). \label{eq:main_constraint}
\end{align}

\begin{theorem}
    \label{thm:main}
    Let \(H\) be matrix concave and \(0 < \gamma \leq 1\), and suppose that for all \(x \in \R^n\) there exists \(u \in \R^m\) satisfying \eqref{eq:main_constraint}.
    Then the following hold.
    \begin{enumerate}
        \item[(i)] For each \(x_k\), \eqref{eq:main_constraint} is a convex matrix inequality in \(u\), and the safety filter
        \begin{subequations}
        \label{eq:main_filter}
        \begin{align}
            u^*(x_k) = \underset{u \in \R^m}{\arg\min} \ & \nrm{u - u_{\rm nom}(x_k)}_2 \\
            \text{s.t.} \ & \eqref{eq:main_constraint}
        \end{align}
        \end{subequations}
        is a convex program.
        \item[(ii)] Any control law satisfying \eqref{eq:main_constraint} for all \(k \geq 0\) renders \(\SSS^{(j)}\) forward invariant.
        \item[(iii)] Along any such trajectory, for all \(i \in \{j, \ldots, p\}\) and all \(k \geq 0\),
        \begin{align}
            \lambda_i(H(x_k)) \geq (1-\gamma)^k \lambda_i(H(x_0)). \label{eq:eig_bound}
        \end{align}
    \end{enumerate}
\end{theorem}
\begin{proof}
    (i) The map \(u \mapsto Y(x_k)^\intercal H(f(x_k) + g(x_k) u) Y(x_k)\) is matrix concave by Section \ref{sec:prelim}, and the right-hand side of \eqref{eq:main_constraint} is constant in \(u\).
    (ii) By Lemma \ref{lem:Phi_safe_subset} and the standing feasibility assumption, \(\Phi_H\) is an SDTE-MCBF for \(\SSS^{(j)}\) on \(\R^n\), and the claim follows from Theorem \ref{thm:SDTE_invariance}.
    (iii) Write \eqref{eq:main_constraint} at step \(k\) as \(Y_k^\intercal H(x_{k+1}) Y_k = (1-\gamma)\Lambda_k + B_k\) with \(B_k \succeq 0\).
    By Lemma \ref{lem:weyl}, \(\lambda_i(Y_k^\intercal H(x_{k+1}) Y_k) \geq (1-\gamma) \lambda_i(\Lambda_k) = (1-\gamma) \lambda_{i+j-1}(H(x_k))\) for \(i \in \{1,\ldots,r\}\).
    By Lemma \ref{lem:interlacing}, \(\lambda_i(Y_k^\intercal H(x_{k+1}) Y_k) \leq \lambda_{i+p-r}(H(x_{k+1})) = \lambda_{i+j-1}(H(x_{k+1}))\).
    Combining and reindexing gives \(\lambda_i(H(x_{k+1})) \geq (1-\gamma)\lambda_i(H(x_k))\) for \(i \in \{j,\ldots,p\}\), and \eqref{eq:eig_bound} follows by induction.
\end{proof}

\begin{remark}[Zeroing]
    Theorem \ref{thm:main}(iii) does not require \(x_0 \in \SSS^{(j)}\).
    If \(\lambda_j(H(x_0)) < 0\), then \eqref{eq:eig_bound} shows that \(\lambda_j(H(x_k))\) converges exponentially to zero from below, so that \(x_k\) approaches \(\SSS^{(j)}\) in the sense of the \(j\)th eigenvalue.
\end{remark}

\begin{remark}[Switching]
    The selected subspace \(Y(x_k)\) changes with \(x_k\), and in the diagonal case \(r = 1\) the selected index \(i^* = \arg\max_i h_i(x_k)\) may switch between consecutive steps.
    Forward invariance is established pointwise at each step and is unaffected by switching.
    Switching affects only the conservatism of \eqref{eq:main_constraint}, since the trajectory is confined to the currently selected convex subset until the selection changes.
\end{remark}

\begin{remark}[Relation to \cite{usevitch2026dtmcbf}]
    In \cite{usevitch2026dtmcbf}, safe subset functions are constructed by projection onto convex components of the unsafe set, and the question of whether other constructions exist is left open.
    The function \(\Phi_H\) in \eqref{eq:Phi_H} is a safe subset function obtained without projection, and Theorem \ref{thm:main} shows that the SDTE-MCBF framework applies to it directly.
\end{remark}

\subsection{Comparison with a Full-Matrix Condition}
\label{ssec:comparison}

An alternative sufficient condition for forward invariance of \(\SSS^{(j)}\) is obtained by shifting the whole matrix rather than a selected compression.

\begin{proposition}
    \label{prop:weyl_condition}
    Let \(0 < \gamma \leq 1\).
    If \(u_k\) satisfies
    \begin{align}
        H(f(x_k) + g(x_k) u_k) \succeq H(x_k) - \gamma \lambda_j(H(x_k)) I_p \label{eq:weyl_constraint}
    \end{align}
    for all \(k \geq 0\), then \(\lambda_j(H(x_k)) \geq (1-\gamma)^k \lambda_j(H(x_0))\) for all \(k \geq 0\), and \(\SSS^{(j)}\) is forward invariant.
    Moreover, every \(u_k\) satisfying \eqref{eq:weyl_constraint} satisfies \eqref{eq:main_constraint}.
\end{proposition}
%
%
\begin{proof}
    Write \eqref{eq:weyl_constraint} as \(H(x_{k+1}) = H(x_k) - \gamma \lambda_j(H(x_k)) I_p + B_k\) with \(B_k \succeq 0\).
    By Lemma \ref{lem:weyl}, \(\lambda_j(H(x_{k+1})) \geq \lambda_j(H(x_k)) - \gamma \lambda_j(H(x_k)) = (1-\gamma)\lambda_j(H(x_k))\), and the first claim follows by induction.
    For the second claim, compress \eqref{eq:weyl_constraint} by \(Y_k = Y(x_k)\) to obtain \(Y_k^\intercal H(x_{k+1}) Y_k \succeq \Lambda_k - \gamma \lambda_j(H(x_k)) I_r\).
    Since \(\Lambda_k \succeq \lambda_j(H(x_k)) I_r\), it holds that \(\Lambda_k - \gamma \lambda_j(H(x_k)) I_r \succeq \Lambda_k - \gamma \Lambda_k = (1-\gamma)\Lambda_k\), which is \eqref{eq:main_constraint}.
\end{proof}

Proposition \ref{prop:weyl_condition} shows that \eqref{eq:weyl_constraint} is at least as restrictive as \eqref{eq:main_constraint}.
%
Constraint \eqref{eq:weyl_constraint} acts on all \(p\) eigenvalues of \(H(x_{k+1})\).
The \(j-1\) smallest of these eigenvalues do not affect membership in \(\SSS^{(j)}\).
Constraint \eqref{eq:main_constraint} acts only on the \(r\)-dimensional compression that determines membership.
In addition, \eqref{eq:weyl_constraint} is convex in \(u\) only if \(H\) is matrix concave, the same condition required by Theorem \ref{thm:main}, so the relaxation comes at no cost in convexity.

\subsection{Specialization to the Sets of Section \ref{sec:obstacles}}
\label{ssec:specialization}

For the three cases of Section \ref{sec:obstacles}, \(j = p\), \(r = 1\), and \(Y(x_k)\) is a unit eigenvector \(v_k\) of \(H(x_k)\) associated with \(\lambda_p(H(x_k))\).
The constraint \eqref{eq:main_constraint} becomes the scalar inequality
\begin{align}
    v_k^\intercal H(f(x_k) + g(x_k) u)\, v_k \geq (1-\gamma)\, \lambda_p(H(x_k)). \label{eq:scalar_constraint}
\end{align}

\emph{Polytope.}
With \(H_\mathrm{P}\) in \eqref{eq:H_polytope}, \(v_k = e_{i^*}\) with \(i^* = \arg\max_i (a_i^\intercal x_k - b_i)\), and \eqref{eq:scalar_constraint} is the single linear inequality
\begin{align}
    a_{i^*}^\intercal \pth{f(x_k) + g(x_k) u} - b_{i^*} \geq (1-\gamma)\pth{a_{i^*}^\intercal x_k - b_{i^*}}.
\end{align}

\emph{Spectrahedron.}
With \(H_\mathrm{Sp}\) in \eqref{eq:H_spectrahedron}, \(v_k\) is the leading eigenvector of \(-(A_0 + \sum_i x_{k,i} A_i)\), and \eqref{eq:scalar_constraint} is linear in \(u\) since \(H_\mathrm{Sp}\) is affine.

\emph{Union of convex sets.}
With \(H_\cup\) in \eqref{eq:H_union}, \(v_k = e_{i^*}\) with \(i^* = \arg\max_i h_i(x_k)\), and \eqref{eq:scalar_constraint} is \(h_{i^*}(f(x_k) + g(x_k)u) \geq (1-\gamma) h_{i^*}(x_k)\), equivalently the linear matrix inequality
\begin{align}
    H_{i^*}\pth{f(x_k) + g(x_k) u} \succeq (1-\gamma)\, h_{i^*}(x_k)\, I_{p_{i^*}}.
\end{align}

In the first two cases the safety filter \eqref{eq:main_filter} is a quadratic program with one linear constraint; in the third it is a semidefinite program of the same size as the filter for the single convex set \(\SSS_{i^*}\).
In all three cases the filter requires one eigenvector computation of a \(p \times p\) matrix per step in addition to the convex program.

\section{Numerical Simulations}
\label{sec:simulations}


\subsection{Setup}

Consider two discretized double integrators in the plane,
\begin{align}
    x_{k+1} = A x_k + B u_k, \label{eq:dbl_int}
\end{align}
where \(x_k \triangleq \matl p_{\rmh,k} & v_{\rmh,k} & p_{\rmv,k} & v_{\rmv,k} \matr^\intercal\) collects horizontal and vertical positions and velocities, \(u_k \in \R^2\) is the acceleration command, \(A \triangleq \operatorname{diag}(A_\rmd, A_\rmd)\), \(B \triangleq \operatorname{diag}(B_\rmd, B_\rmd)\),
\begin{align}
    A_\rmd \triangleq \matl 1 & T_\rms \\ 0 & 1 \matr, \quad B_\rmd \triangleq \matl T_\rms^2/2 \\ T_\rms \matr,
\end{align}
and \(T_\rms = 0.01\) s.
Let \(C_{\rm pos} \triangleq \matl 1 & 0 & 0 & 0 \\ 0 & 0 & 1 & 0 \matr\) extract the position \(p_k = C_{\rm pos} x_k\).
The nominal controller \(u_{\rm nom}(x_k)\) is a reference-tracking LQR controller for the constant reference \(r_k \equiv \matl 8 & 0 & 8 & 0 \matr^\intercal\), and \(x_0 = 0\).
Details of the LQR design are omitted since the focus is on the safety filter.

Two obstacles are placed between \(x_0\) and the reference.

\begin{itemize}
    \item A regular pentagon \(\overline{\SSS}_\mathrm{P}\) with center \(\matl 3.5 & 2.5 \matr^\intercal\) m and circumradius 2 m, written in the halfspace form of Section \ref{ssec:polytope} with \(p = 5\), \(a_i \in \R^2\), \(b_i \in \R\).
    \item A spectrahedron \(\overline{\SSS}_\mathrm{Sp}\) with center \(\matl 4 & 6 \matr^\intercal\) m, \(p = 3\), and
    \footnotesize
    \begin{align*}
        A_0 = 2 I_3, \quad A_1 &= \matl 0 & 0.8 & 0 \\ 0.8 & 0 & 0.8 \\ 0 & 0.8 & 0 \matr,\\
        A_2 &= \matl 0 & 0 & 1.6 \\ 0 & 0 & 2.4 \\ 1.6 & 2.4 & 0 \matr,
    \end{align*}
    \normalsize
    so that \(\overline{\SSS}_\mathrm{Sp} = \{ p \in \R^2 : A_0 + (p_1 - 3.5) A_1 + (p_2 - 3.5) A_2 \succeq 0 \}\).
\end{itemize}
Since the obstacles depend only on position, the functions of Section \ref{sec:obstacles} are evaluated at \(C_{\rm pos} x\); with a slight abuse of notation, \(H_\mathrm{P}(x)\) and \(H_\mathrm{Sp}(x)\) below denote \(H_\mathrm{P}(C_{\rm pos} x)\) and \(H_\mathrm{Sp}(C_{\rm pos} x)\).
A buffer is added to each obstacle: for the pentagon, \(b_i\) is replaced by \(b_i + \epsilon_\mathrm{P}\) with \(\epsilon_\mathrm{P} = 0.2\); for the spectrahedron, \(p - p_\mathrm{Sp}\) is replaced by \((1 - \epsilon_\mathrm{Sp})(p - p_\mathrm{Sp})\) with \(\epsilon_\mathrm{Sp} = 0.2\), which scales the obstacle about its center.

The two obstacles are handled as separate constraints, one for each \(H\), in the same filter, so that the overall safe set is \(\SSS_\mathrm{P} \cap \SSS_\mathrm{Sp}\), a conjunction of two disjunctive sets.
The filters below apply Theorem \ref{thm:main} to each \(H\) independently, which preserves forward invariance of the intersection.

\subsection{Safety Filters}

Three safety filters are compared, each solving \(\min_u \nrm{u - u_{\rm nom}(x_k)}_2\) subject to one constraint per obstacle, with \(\gamma = 0.2\) and \(f(x_k) = A x_k\), \(g(x_k) = B\).
\begin{itemize}
    \item \textbf{Selected (proposed).} Constraint \eqref{eq:scalar_constraint} for each obstacle, with \(v_k\) the leading eigenvector of \(H_\mathrm{P}(x_k)\), respectively \(H_\mathrm{Sp}(x_k)\).
    Both constraints are linear in \(u\), and the filter is a quadratic program with two linear constraints, solved with \texttt{lsqlin}.
    %
    \item \textbf{Full-matrix.} Constraint \eqref{eq:weyl_constraint} is modified and applied for each obstacle.
    Although \eqref{eq:weyl_constraint} enforces the constraint, in practice it can cause the system to to become trapped near an obstacle.
    Hence, \eqref{eq:weyl_constraint} is replaced with
    \small
    \begin{align}
        H(f(x_k) + g(x_k) u_k)  \succeq & (1 + c_{\rm perp})H(x_k) \nn \\
        & - (\gamma + c_{\rm perp}) \lambda_j(H(x_k)) I_p, \label{eq:weyl_constraint_2}
    \end{align}
    \normalsize
    where $c_{\rm perp} > 0$ enables the system to move along the obstacle boundary in a direction perpendicular to the normal at the closest boundary point.
    For this example, we set $c_{\rm perp} > 0.05.$
    For the pentagon this is five linear inequalities; for the spectrahedron it is a \(3 \times 3\) linear matrix inequality, and the filter is a semidefinite program solved with \texttt{MOSEK} through \texttt{YALMIP}.
\end{itemize}
The full-matrix filter is included to quantify the conservatism established in Proposition \ref{prop:weyl_condition}.
%
%
All simulations run for \(k \in \{0, \ldots, 1200\}\) on a laptop PC running Windows 10 Pro, version 22H2, OS build 19045.7725, with an Intel Core i7-10750H processor running at 2.60 GHz and a 32GB, 3200 MHz RAM, with MATLAB version R2024b Update 5.
%

\subsection{Results}


Figure \ref{fig:traj} shows the trajectories under the three filters and Figure \ref{fig:eigs} shows \(\lambda_p(H_\mathrm{P}(x_k))\) and \(\lambda_p(H_\mathrm{Sp}(x_k))\) along each.
Table \ref{tab:runtime} reports the runtime per step.
%
The following observations follow from the results.
\begin{itemize}
    \item \textbf{Forward invariance:} Both eigenvalue traces remain nonnegative for the selected and full-matrix filters, which confirms the forward invariance of both filters.
    Note that the eigenvalue traces approach zero as the system states approach the boundaries of the safe sets.
    %
    \item \textbf{Optimization solutions:} The trajectories obtained with the selected and full-matrix filters nearly coincide at every step, such that the difference between their position trajectories has a mean of $\matl 2.8458 \cdot 10^{-2} & 4.4482 \cdot 10^{-2} \matr^\rmT$ and standard devition of $\matl 3.8847 \cdot 10^{-2} & 7.0639 \cdot 10^{-2} \matr^\rmT.$
    \item \textbf{Runtime:} The selected filter QP is more than an order of magnitude faster than the full-matrix SDP.
\end{itemize}

\begin{figure}[!ht]
    \vspace{-1em}
    \centering
    \includegraphics[width=\columnwidth]{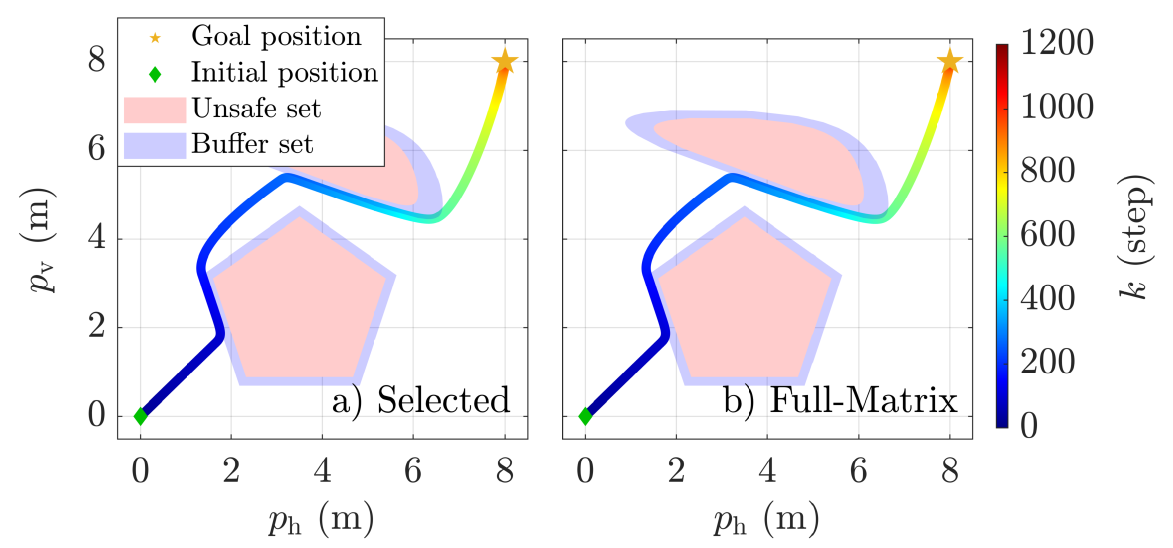}
    \caption{Trajectories under the a) selected and b) full-matrix safety filters. Obstacles are shown with their buffered boundaries.}
    \label{fig:traj}
\end{figure}

\begin{figure}[!ht]
    \vspace{0.5em}
    \centering
    \includegraphics[width=\columnwidth]{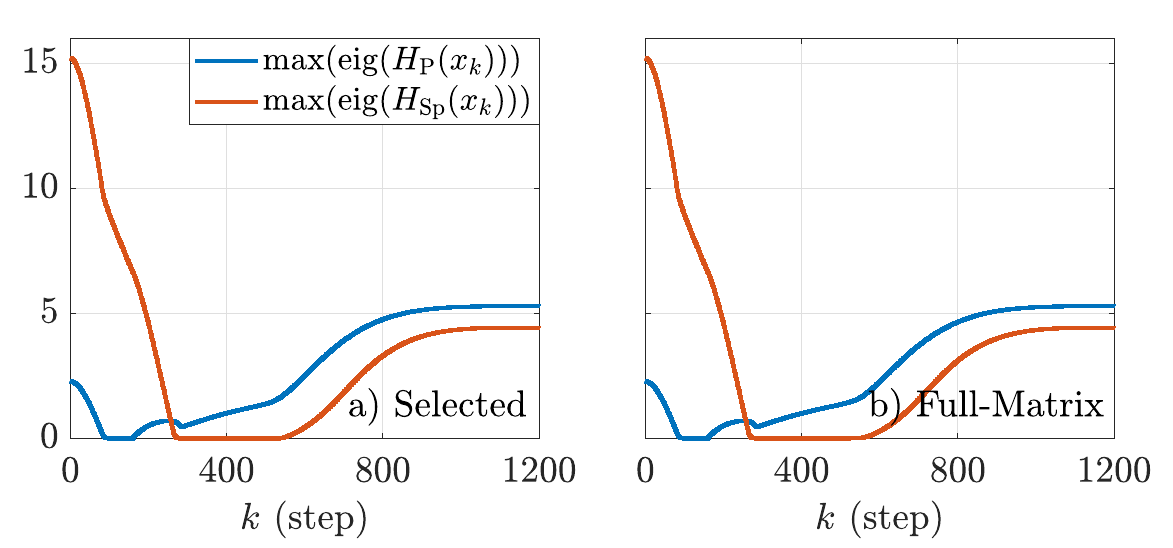}
    \caption{Largest eigenvalue of \(H_\mathrm{P}(x_k)\) and \(H_\mathrm{Sp}(x_k)\) along the trajectories of Figure \ref{fig:traj}  under the a) selected and b) full-matrix safety filters. Nonnegativity of each trace is forward invariance of the corresponding safe set.}
    \label{fig:eigs}
\end{figure}


\begin{table}[!ht]
    \vspace{-1em}
    \centering
    \caption{Runtime per step for the two safety filters.}
    \begin{tabular}{lcc}
    \toprule
    Filter & Mean (s) & Max (s)\\
    \midrule
    Selected (QP) & 1.9001 $\cdot$ 10$^{-3}$ & 5.5191 $\cdot$ 10$^{-2}$ \\
    Full-matrix (SDP) & 9.2078 $\cdot$ 10$^{-2}$ & 6.0094 $\cdot$ 10$^{-1}$ \\
    \bottomrule
    \end{tabular}
    \label{tab:runtime}
\end{table}

\section{Conclusions}
\label{sec:conclusions}

This paper considered discrete-time safety filtering when the safe set is a union of convex sets and the discrete-time constraint is nonconvex in the input.
When the matrix function is matrix concave, the safe set was shown to be a union of convex sets, and complements of polytopes, complements of spectrahedra, and unions of convex sets were written in this form.
Selecting eigenvectors of the matrix function at the current state yields a convex input constraint for any number of required nonnegative eigenvalues, using one eigendecomposition and one convex program per step and no projection onto the unsafe set.
Along closed-loop trajectories, every eigenvalue that determines membership in the set satisfies a geometric lower bound.
The condition that bounds the entire matrix implies the proposed constraint, so the proposed constraint admits at least as many inputs at each state.
In simulation with a pentagon and a spectrahedron obstacle, both the proposed filter and the full-matrix filter kept both safe sets invariant, and their position trajectories differed by centimeters.
The mean solve time per step was \(1.9\) ms for the proposed quadratic program and \(92\) ms for the full-matrix semidefinite program.
The two filters were solved with different solver stacks, so the ratio does not isolate the effect of the constraint.
Both obstacles had one required nonnegative eigenvalue, so the case of several required eigenvalues is covered by the theory and not by the simulations.
Joint feasibility of the constraints from several obstacles is assumed and not established.

Future work will address joint feasibility for multiple obstacles, simulations with several required nonnegative eigenvalues, and applications to unmanned aerial vehicles and legged robots under noise and uncertainty.

\section*{Acknowledgements}

The authors also thank and acknowledge AI tools, including ChatGPT and Claude, for reviewing and revising the manuscript for technical and grammatical errors, critically examining the manuscript, and suggesting corrective actions.

\printbibliography

\end{document}

%% file: PackagesCommands.tex
\usepackage{amsmath} 
\usepackage{amsfonts}
\usepackage{amsthm}

\newtheorem{theorem}{Theorem}

\newtheorem{proposition}{Proposition}[section]
\theoremstyle{definition}

\newtheorem{remark}{Remark}
\newtheorem{lemma}{Lemma}[section]

\usepackage{float} 
\usepackage[skip=0em,font=footnotesize]{caption}

\usepackage{tikz}
\usetikzlibrary{shapes,arrows,calc,positioning}
\tikzstyle{bigblock} = [draw, fill=blue!20, rectangle, 
    minimum height=6em, minimum width=8em]
\tikzstyle{medblock} = [draw, fill=blue!20, rectangle, 
    minimum height=4em, minimum width=4em]    
\tikzstyle{mux} = [draw, fill=black!20, rectangle, 
    minimum height=5em, minimum width=0.1em]    
\tikzstyle{smallblock} = [draw, fill=blue!20, rectangle, 
    minimum height=2em, minimum width=3em]
    
\tikzstyle{data_block} = [draw, fill=green!20, rectangle, 
    minimum height=2em, minimum width=3em]
\tikzstyle{ops_block} = [draw, fill=blue!20, rectangle, 
    minimum height=2em, minimum width=3em]    
\tikzstyle{est_block} = [draw, fill=red!20, rectangle, 
    minimum height=2em, minimum width=3em]    
    
\tikzstyle{sum} = [draw, fill=blue!20, circle, node distance=1cm,minimum height=0.5cm]
\tikzstyle{signal} = [coordinate]
\tikzstyle{pinstyle} = [pin edge={to-,thin,black}]
\tikzstyle{block} = [draw, fill=blue!20, rectangle, 
    minimum height=3em, minimum width=9em]
\tikzstyle{blockS} = [draw, fill=blue!20, rectangle, 
    minimum height=3em, minimum width=4em]    
\tikzstyle{input} = [coordinate]
\tikzstyle{output} = [coordinate]
\usetikzlibrary{matrix}

\usetikzlibrary{positioning}
\usetikzlibrary{math} 

\usetikzlibrary{calc,arrows.meta}
\definecolor{safe}{RGB}{31,119,180}

\usepackage{hyperref}
\usepackage{xcolor}
\hypersetup{
    colorlinks,
    linkcolor={blue!100!black},
    citecolor={blue!50!black},
    urlcolor={blue!80!black}
}
\usepackage{adjustbox}

\newcommand{\bc}{\begin{center}}
\newcommand{\ec}{\end{center}}
\newcommand{\benum}{\begin{enumerate}}
\newcommand{\eenum}{\end{enumerate}}
\newcommand{\nn}{\nonumber}
\newcommand{\matl}{\left[ \begin{array}}
\newcommand{\matr}{\end{array} \right]}

\renewcommand{\matl}{\begin{bmatrix}}
\renewcommand{\matr}{\end{bmatrix}}

\newcommand{\matls}{\left[ \begin{smallmatrix}}
\newcommand{\matrs}{\end{smallmatrix} \right]}

\newcommand{\dx}{{\rm d}x}

\newcommand{\rmT}{{\rm T}}

\newcommand{\rmd}{{\rm d}}

\newcommand{\rmh}{{\rm h}}

\newcommand{\rms}{{\rm s}}

\newcommand{\rmv}{{\rm v}}

\newcommand{\SC}{{\mathcal C}}

\newcommand{\SSS}{{\mathcal S}}

\let\labelindent\relax
\usepackage{enumitem,amssymb}
\newlist{todolist}{itemize}{2}
\setlist[todolist]{label=$\square$}
\usepackage{pifont}
\newtheorem{define}{Definition}

\newcommand{\R}{\mathbb R}

\newcommand{\pth}[1]{\left(#1\right)} 
\newcommand{\nrm}[1]{\left \lVert#1\right \rVert} 

\DeclarePairedDelimiter{\ceil}{\lceil}{\rceil}
\DeclarePairedDelimiter{\floor}{\lfloor}{\rfloor}
\DeclarePairedDelimiter{\abs}{\lvert}{\rvert}

\makeatletter
\let\oldceil\ceil
\def\ceil{\@ifstar{\oldceil}{\oldceil*}}
\makeatother
\makeatletter
\let\oldfloor\floor
\def\floor{\@ifstar{\oldfloor}{\oldfloor*}}
\makeatother
\makeatletter
\let\oldnorm\norm
\def\norm{\@ifstar{\oldnorm}{\oldnorm*}}
\makeatother
\makeatletter
\let\oldabs\abs
\def\abs{\@ifstar{\oldabs}{\oldabs*}}
\makeatother

%% file: Figures/Fig_obstacles_body.tex
\begin{adjustbox}{width=\columnwidth}
\begin{tikzpicture}[scale=1, >={Stealth[length=4pt]}, font=\footnotesize]

\def\L{2.0}      
\def\dx{4.6}     
\def\dy{4.6}     

\def\dd{0.809017}
\newcommand{\halfspace}[3]{
  \fill[#3] ($ (#1:#2) + (#1+90:4) $) -- ($ (#1:#2) + (#1-90:4) $)
            -- ($ (#1:#2) + (#1-90:4) + (#1:4) $) -- ($ (#1:#2) + (#1+90:4) + (#1:4) $) -- cycle;
}
\newcommand{\edgeline}[3]{
  \draw[#3] ($ (#1:#2) + (#1+90:4) $) -- ($ (#1:#2) + (#1-90:4) $);
}

\begin{scope}
  \clip (-\L,-\L) rectangle (\L,\L);
  \foreach \k in {0,...,4}{ \halfspace{126+72*\k}{\dd}{safe, opacity=0.13} }
  \foreach \k in {0,...,4}{ \edgeline{126+72*\k}{\dd}{gray!60, dashed, thin} }
  \draw[thick, fill=white] (90:1) -- (162:1) -- (234:1) -- (306:1) -- (18:1) -- cycle;
  \node at (0,0) {$\overline{\mathcal{S}}_\mathfrak{P}$};
  \node[safe] at (1.35,1.35) {$\mathcal{S}_\mathfrak{P}$};
  \node[font=\scriptsize, rotate=36] at ($ (126:\dd) + (126:0.22) + (216:0.9) $) {$a_1^\intercal x = b_1$};
  \node[anchor=north west, font=\footnotesize\bfseries] at (-\L,\L) {(A)};
\end{scope}

\begin{scope}[xshift=\dx cm]
  \clip (-\L,-\L) rectangle (\L,\L);
  \fill[safe, opacity=0.13] (-\L,-\L) rectangle (\L,\L);
  \def\blob{plot[domain=0:360, samples=90, smooth cycle] (\x:{0.95+0.12*cos(3*\x)})}
  \foreach \t in {0,40,...,320}{
    \pgfmathsetmacro{\r}{0.95+0.12*cos(3*\t)}
    \edgeline{\t}{\r}{gray!60, dashed, thin}
  }
  \draw[thick, fill=white] \blob;
  \node at (0,0) {$\overline{\mathcal{S}}_\mathfrak{Sp}$};
  \node[safe] at (1.35,1.35) {$\mathcal{S}_\mathfrak{Sp}$};
  \node[font=\scriptsize, fill=white, inner sep=1pt] at (-0.85,1.45) {$v^\intercal H(x) v = 0$};
  \node[anchor=north west, font=\footnotesize\bfseries] at (-\L,\L) {(B)};
\end{scope}

\begin{scope}[xshift=2*\dx cm]
  \clip (-\L,-\L) rectangle (\L,\L);
  \begin{scope}[transparency group, opacity=0.13]
    \fill[safe] (-1.8,-0.5) rectangle (0.3,0.5);
    \fill[safe] (-0.3,-1.8) rectangle (0.5,1.2);
    \fill[safe] (0.2,0.6) ellipse (1.5 and 0.7);
  \end{scope}
  \draw[safe, thick] (-1.8,-0.5) rectangle (0.3,0.5);
  \draw[safe, thick] (-0.3,-1.8) rectangle (0.5,1.2);
  \draw[safe, thick] (0.2,0.6) ellipse (1.5 and 0.7);
  \node[safe, font=\scriptsize] at (-1.3,0) {$\mathcal{S}_1$};
  \node[safe, font=\scriptsize] at (0.1,-1.3) {$\mathcal{S}_2$};
  \node[safe, font=\scriptsize] at (1.2,0.6) {$\mathcal{S}_3$};
  \node at (-1.2,-1.3) {$\overline{\mathcal{S}}_\cup$};
  \node[anchor=north west, font=\footnotesize\bfseries] at (-\L,\L) {(C)};
\end{scope}

\end{tikzpicture}
\end{adjustbox}